\documentclass{article}
\usepackage{graphicx} 
\usepackage{amsmath, amssymb, amsthm}
\usepackage{xcolor}
\usepackage{comment}

\usepackage{algorithm}
\usepackage{algpseudocode}

\usepackage[maxnames=50]{biblatex}
\usepackage{hyperref}

\newcommand{\F}{\mathbb{F}}
\newcommand{\Z}{\mathbb{Z}}
\newcommand{\R}{\mathbb{R}}

\newcommand{\covol}{\operatorname{covol}}
\newcommand{\OD}{\operatorname{OD}}
\newcommand{\wt}{\operatorname{wt}}
\newcommand{\im}{\operatorname{im}}

\newtheorem{theorem}{Theorem}

\newtheorem{lemma}[theorem]{Lemma}

\newtheorem{corollary}[theorem]{Corollary}
\newtheorem{problem}[theorem]{Problem}

\theoremstyle{definition}

\newtheorem{definition}[theorem]{Definition}
\theoremstyle{remark}
\newtheorem{remark}[theorem]{Remark}
\newtheorem{example}[theorem]{Example}

\title{Polynomial Lattices for the BIKE Cryptosystem}
\author{Michael Schaller \\ E-mail: michael.schaller@math.uzh.ch}
\date{February 23, 2026}

\begin{document}

\maketitle

\begin{abstract}
In this paper we introduce a rank $2$ lattice over a polynomial ring arising from the public key of the BIKE cryptosystem \cite{aragon2022bike}.
The secret key is a sparse vector in this lattice.
We study properties of this lattice and generalize the recovery of weak keys from \cite{BardetDLO16}.
In particular, we show that they implicitly solved a shortest vector problem in the lattice we constructed.
Rather than finding only a shortest vector, we obtain a reduced basis of the lattice which makes it possible to check for more weak keys.
\end{abstract}

\section{Introduction and Related Work}

One of the round $4$ candidates of the NIST competition for Post-Quantum cryptography was the BIKE cryptosystem \cite{aragon2022bike}.
It has not been selected for standardization, but the related scheme HQC \cite{melchor2018hamming} will be standardized.

In this paper we construct a lattice over a polynomial ring to study the BIKE cryptosystem.
The lattice is constructed from the public key while the secret key is a sparse vector in this lattice.
We investigate the properties of this lattice, in particular the successive minima of the lattice.

Moreover, in our framework we can reinterpret the work from \cite{BardetDLO16} in which a class of weak keys for BIKE has been exhibited.
We extend their ideas using the work of Mahler \cite{mahler1941analogue} and Lenstra \cite{lenstra1983factoring} on lattices over polynomial rings and reformulate the results of \cite{BardetDLO16}.
This enables us to give a more general perspective on the problem and to obtain more weak keys than in the original work in \cite{BardetDLO16}.
It turns out that (without using this perspective) they find a shortest vector in the lattice.
Using the lattice we are able to obtain a reduced basis instead of just a shortest vector in the lattice which allows to test for more weak keys.

We would like to point out that the weak keys from \cite{BardetDLO16} are not related to the weak keys that are studied in connection with decoding failure rates.

Our framework bears similarity to the function field decoding approach introduced in \cite{BombarCD22}.
While they work on different problems, in particular a search to decision reduction, several of the general ideas about the BIKE cryptosystem and more generally quasi-cyclic codes are implicit in their work.
Thinking about the Minkowski embedding for rings of integers of number fields and classical lattices in the real numbers makes the connection to their work more obvious.

There have also been other works that investigated lattices over non-archimedean fields in other settings.
Especially noteworthy are the works of Bernstein \cite{Bernstein2008smallheight}, Cohn and Heninger \cite{CohnH15} who used an analogy with Coppersmith's theorem in order to reinterpret the Guruswami-Sudan list decoding algorithm.
In addition, there have been applications to the computations of Riemann-Roch spaces, for example in Bauch's PhD thesis \cite{Bauch2014}.

The structure of the paper is as follows.
In Section \ref{sect:BIKE} we briefly describe the BIKE cryptosystem and the weak key attack by the authors of \cite{BardetDLO16}.
In Section \ref{sect:real_lattices} we deal with lattices over $\Z$ and give some motivation for the use of polynomial lattices for BIKE later.
In Section \ref{sect:poly_lattices} we introduce and prove results needed later about lattices in the non-archimedean case.
Section \ref{sect:BIKE_Lattices} contains the main part of the paper in which we introduce a lattice related to the BIKE cryptosystem and prove all the statements we need for the lattices in our treatment of the weak key attack on the BIKE cryptosystem.
We conclude with Section \ref{sect:conclusion}.

\section{The BIKE Cryptosystem and a Class of Weak Keys}
\label{sect:BIKE}

First of all we describe the BIKE cryptosystem \cite{aragon2022bike}.
Let us fix a prime $r$ such that $2$ generates $\F_r^\times$.
The secret key consists of two sparse (weight $v \sim \sqrt{r}$) elements $h_1, h_2 \in \F_2[x]/(x^r - 1)$.
The public key is $\frac{h_1}{h_2}$.
The goal of a key recovery attack is to recover $h_1$ and $h_2$ from $h$.
Actually, we do not need the very same $h_1$ and $h_2$.
Any two sparse elements $h_1, h_2$ satisfying $h = \frac{h_1}{h_2}$ are good enough.

Next we sketch the idea of the attack from \cite{BardetDLO16}.
Let us rewrite the equation for the public key as
\[
    h_1 = h h_2 \mod x^r - 1.
\]
This gives us the rational reconstruction problem which can be solved uniquely by the extended Euclidean algorithm as long as $\deg(h_1) + \deg(h_2) < r$.
The proof of uniqueness is in essence the same as the proof of Theorem 4.8 in \cite{shoup2009computational}.
The authors of \cite{BardetDLO16} also use group actions that come from the code being quasi-cyclic and from the squaring map on $\F_2[x]/(x^r - 1)$ to get more weak keys.
We will not describe it here, but only mention that one can do the same thing using the new approach from our paper.

More information about the rational reconstruction problem in the case of integers can be found in \cite{shoup2009computational}.

In the next sections we will give another perspective on the problem which we deem more natural and which allows to change the condition that $\deg(h_1) + \deg(h_2) < r$.

\section{Lattices and Thue's Lemma}
\label{sect:real_lattices}

In this section we motivate the use of lattices in the later sections.
The results here give an indication of what will happen in Sections \ref{sect:poly_lattices} and \ref{sect:BIKE_Lattices} for polynomial lattices.
Here we consider lattices as discrete subgroups of $\R^n$ (in contrast to the next sections where we consider lattices over polynomial rings).
All the lattices are assumed to be full rank.

Let $\gamma \in \Z_{> 0}$.
Now we introduce $\gamma$-ary lattices (usually called $q$-ary lattices where $\gamma$ is replaced by $q$).
Let $k \leq n$ and $A \in \Z/(\gamma)^{k \times n}$ be a matrix.
We define
\[
    \Lambda_\gamma^\perp(A) = \{x \in \Z^n : Ax = 0 \mod \gamma \}.
\]
We get the following well known lemma.
\begin{lemma}
\label{lemma:covol q-ary lattice}
Let $k \leq n$ and $A \in \Z/(\gamma)^{k \times n}$ be a matrix.
Then
    \[
        \covol(\Lambda_\gamma^\perp(A)) \leq \gamma^k.
    \]
\end{lemma}
\begin{proof}
    Let $\pi: \Z^n \to \Z/(\gamma)^n$ be the natural projection.
    Then
    \[
        \Lambda_\gamma^\perp(A) = \pi^{-1} (\ker A).
    \]
    Hence 
    \[
        \Z^n / \Lambda_\gamma^\perp(A) \cong (\Z/(\gamma)^n)/(\ker A) \cong \im A \subseteq \Z/(\gamma)^k.
    \]
    The last inclusion implies the claim.
\end{proof}

We first state Thue's lemma which can be found for instance in \cite{shoup2009computational} and then prove it using Minkowski's Theorem.

\begin{lemma}[Thue]
	\label{lemma:Thue}
	Let $\gamma, b, a^\ast, t^\ast \in \Z_{\geq 0}$.
	Assume $0 < a^\ast \leq \gamma < a^\ast t^\ast$.
	Then there exist $a, t \in \Z$ such that $|a| < a^\ast, 0 < |t| < t^\ast$ satisfying
	\[
	a \equiv b t \mod \gamma.
	\]
\end{lemma}
\begin{proof}
	Consider the following lattice:
	\[
	L = \{(a, t) : a - b t \equiv 0 \mod \gamma \}.
	\]
	We want to estimate the covolume of the lattice in order to apply Minkowski's Theorem.
	We have the trivial inclusions $\gamma \Z^2 \subseteq L \subseteq \Z^2$.
	But more is true.
	$L = \{(a, t) \in \Z^2 : (1, -b) (a, t)^T \equiv 0 \mod \gamma \}$ is an $\gamma$-ary lattice.
	Thus, using Lemma \ref{lemma:covol q-ary lattice}, we get  that
	\[
	\covol(L) \leq \gamma.
	\]
	Under the assumption that the volume of a convex, symmetric (around $0$) body $R$ is larger than $4\gamma$, we get, using Minkowski's Theorem, a nonzero lattice point in $R$.
	The rectangle $R = \{(x_1, x_2) : |x_1| < a^\ast, |x_2| < t^\ast \}$ has volume $4 a^\ast t^\ast > 4 \gamma$.
	Hence there exists a nonzero lattice point in it.
	It is easy to see that if $t = 0$ we also get $a = 0$ and the statement is proven.
\end{proof}
Note that we could also additionally assume $a \geq 0$ since we will get a nonzero lattice point $x$ and its inverse $-x$.
In \cite{shoup2009computational} the author gives an efficient algorithm relying on the extended Euclidean algorithm to give an explicit solution for Lemma \ref{lemma:Thue}.
Note that this amounts to finding a short vector in a rank $2$ lattice.
So instead of the extended Euclidean algorithm we could run a lattice reduction algorithm like the Lagrange algorithm for rank $2$ lattices.

\section{Lattices over Polynomial Rings}
\label{sect:poly_lattices}

In this section we explore the concept of a lattice over a polynomial ring $k[x]$ where $k$ is a field.
The study of this subject has been initiated by Mahler in \cite{mahler1941analogue} in which he treated an analogue of Minkowski's geometry of numbers \cite{minkowski1910geometrie} for non-archimedean fields.
He also proved an analogue of Minkowski's Theorem.
We will mainly be concerned with the paper by Lenstra \cite{lenstra1983factoring} in which he presented a polynomial time algorithm for lattice reduction which ends up with a basis that has an orthogonality defect of $0$.
We will define lattices only over polynomial rings and not over general non-archimedean fields as done in \cite{mahler1941analogue}.
For the purposes of this paper it is not even necessary to go to the field of rational functions or to Laurent series.
We will also not go into the improvements of lattice reduction algorithms over polynomial rings like \cite{MuldersS03, GiorgiJV03} since for our purposes Lenstra's algorithm is sufficient.

We mainly recall the relevant results from \cite{lenstra1983factoring}, state the analogue of Minkowski's Theorem due to Mahler \cite{mahler1941analogue} and add a few results we could not find in the literature.

\begin{definition}
    For $f \in \F_q[x]$ we define $|f| = \deg f$ for $f \neq 0$ and $|0| = - \infty$.
    We define the \textit{norm} of $a = (a_1, \ldots, a_n) \in \F_q[x]^n$ as 
    \[
        |(a_1, \ldots, a_n)| = \max \{|a_1|, \ldots, |a_n| \}.
    \]
\end{definition}
Note that $|f|$ for $f \in \F_q[x]$ does not behave like a non-archimedean norm, but like the additive inverse of a discrete valuation or equivalently like the logarithm of a non-archimedean norm.

\begin{definition}
    Let $b_1, \ldots, b_n \in \F_q[x]^n$ be linearly independent over $\F_q[x]$.
    Then the \textit{lattice} spanned by $b_1, \ldots, b_n$ is defined to be
    \[
        L = \left\{ \sum_{i=1}^n \mu_i b_i : \mu_i \in \F_q[x] \right\}.
    \]
    The \textit{determinant} $\det(L)$ of the lattice is defined as the determinant of the matrix $B$ which has $b_1, \ldots, b_n$ as rows.
    The \textit{covolume} $\covol(L) = |\det(L)|$ of the lattice $L$ is the degree of the determinant.
    The \textit{orthogonality defect} of the basis $b_1, \ldots, b_n$ of the lattice $L$ is defined to be
    \[
        \OD(b_1, \ldots, b_n) = \left(\sum_{i = 1}^n |b_i| \right) - |\det(L)|.
    \]
    We call a basis with orthogonality defect $0$ a \textit{reduced} basis.
\end{definition}

\begin{remark}
    Note that $OD(b_1, \ldots, b_n) \geq 0$.
    Since the determinant of a lattice is unique up to elements in $\F_q^\ast$ we will consider the determinant only up to units.
    Moreover, since we work with the additive inverse of a valuation, we get after taking the orthogonality defect to an appropriate power (say $q^{\OD(b_1, \ldots,b_n)}$) the correct analogue of the usual definition of the orthogonality defect.
\end{remark}

Next we define the analogue of an important classical notion for lattices.
\begin{definition}
\label{def:successive_min_hypercube}
    The first \textit{successive minimum} is defined to be the norm $|m_1|$ of a shortest non-zero vector $m_1$ in the lattice $L$.
    Iteratively we define the $j$-th successive minimum for $2 \leq j \leq n$ as follows.
    Take $j-1$ linearly independent vectors $m_1, \ldots, m_{j-1} \in L$ that attain the corresponding successive minima.
    The $j$-th \textit{successive minimum} is the norm $|m_j|$ of a shortest vector $m_j \in L$ that is linearly independent of $m_1, \ldots, m_{j-1}$.
\end{definition}

\begin{remark}
    The value of the $j$-th successive minimum does not depend on the choice of the previous vectors $m_1, \ldots, m_{j-1}$.
\end{remark}

Mahler defined a more general notion of successive minima with respect to more general convex bodies.
Studying these could be an interesting problem for further research.

Lenstra proved the following result about the relation of successive minima and orthogonality defect.

\begin{theorem}[Lenstra]
    Let $b_1, \ldots, b_n$ be a basis of the lattice $L$ with $OD(b_1, \ldots, b_n) = 0$, ordered in such a way that $|b_1| \leq |b_2| \leq \ldots \leq |b_n|$.
    Then $|b_j|$ is the $j$-th successive minimum and in particular $b_1$ is a shortest vector of the lattice $L$.
\end{theorem}

We will not describe the algorithm from Lenstra's paper, but we state the main result.
Actually, the algorithm is not very long and also simple to implement.

\begin{theorem}[Lenstra]
\label{thm:lenstra complexity}
    Let $C > 0$ be such that $|b_i| \leq C$ for all $i \in \{1, \ldots, n \}$.
    Then there exists an algorithm that takes $O(n^3 C (\OD(b_1, \ldots, b_n) + 1))$ arithmetic operations to transform a basis $b_1, \ldots, b_n$ of $L$ into a reduced basis for $L$.
\end{theorem}

In particular one can find a shortest vector in the given time.
Note that one gets the following corollary using $\OD(b_1, \ldots, b_n) \leq n C$.

\begin{corollary}[Lenstra]
\label{cor:lenstra complexity}
    Let $C > 0$ be such that $|b_i| \leq C$ for all $i \in \{1, \ldots, n \}$.
    Then there exists an algorithm that takes $O(n^4 C^2)$ arithmetic operations to transform a basis $b_1, \ldots, b_n$ of $L$ into a reduced basis for $L$.
\end{corollary}

We will prove the following two lemmata since we could not find the statements and the proofs in the literature.
However, they are not strictly necessary for the remainder of the paper.
It is a classical result that for lattices $L_2 \subseteq L_1 \subseteq \R^n$ we have $[L_1 : L_2] = \frac{\covol(L_2)}{\covol(L_1)}$.
We prove the following analogue.

\begin{lemma}
\label{lemma:index_poly_lattice_covolume}
    Let $L_2 \subseteq L_1 \subseteq \F_q[x]^n$ be two lattices.
    Then
    \[
        [L_1 : L_2] = q^{\left|\frac{\det(L_2)}{\det(L_1)} \right|} = q^{\covol(L_2) - \covol(L_1)}.
    \]
\end{lemma}
\begin{proof}
    By the elementary divisor theorem we can choose a basis $b_1, \ldots, b_n$ of $L_1$ and $\alpha_1, \ldots, \alpha_n \in \F_q[x]^\ast$ such that $\alpha_1 b_1, \ldots, \alpha_n b_n$ is a basis of $L_2$.
    Then by multilinearity of the determinant we get that
    \[
        \det(L_2) = \left(\prod_{i=1}^n \alpha_i \right) \det(L_1).
    \]
    Hence
    \[
        \left|\frac{\det(L_2)}{\det(L_1)} \right| = \sum_{i=1}^n |\alpha_i|.
    \]
    Furthermore,
    \[
        [L_1 : L_2] = \prod_{i=1}^n q^{|\alpha_i|}  = q^{\sum_{i=1}^n |\alpha_i|}.
    \]
    This finishes the proof.
\end{proof}

We define an analogue of $\gamma$-ary matrices in this setting.
Let $g \in \F_q[x]$.
Further assume $k \leq n$ and $A \in \F_q[x]^{k \times n}$.
We define
\[
    \Lambda_g^\perp(A) = \{ a \in \F_q[x]^n : A a \equiv 0 \mod g \}.
\]

We have the following analogue to Lemma \ref{lemma:covol q-ary lattice}.
\begin{lemma}
\label{lemma:index polynomial lattice}
    Let $k \leq n$ and $A \in \F_q[x]^{k \times n}$.
    Then
    \[
        [\F_q[x]^n : \Lambda_g^\perp(A)] \leq q^{|g| k}.
    \]
\end{lemma}
\begin{proof}
    The proof is identical to the proof of Lemma \ref{lemma:covol q-ary lattice}.
\end{proof}

Using Lemma \ref{lemma:index_poly_lattice_covolume} we can also rewrite the previous lemma as $ \covol(\Lambda_g^\perp(A)) \leq |g| k$.

\begin{example}
    Let $A = (I_k| 0_{n-k})$.
    Then
    \[
        \Lambda_g^\perp(A) = (g \F_q[x]^k) \times \F_q[x]^{n-k}.
    \]
    In this case, we get equality in Lemma \ref{lemma:index polynomial lattice}.
    Namely,
    \[
        \covol(\Lambda_g^\perp(A)) = |g| k.
    \]
\end{example}

Finally we will state the theorem by Mahler which is the analogue of Minkowski's Theorem in a slightly generalized form as in \cite{bagshaw2025lattices}.

\begin{theorem}[Mahler]
\label{thm:mahler}
Let $L$ be a lattice and $|m_j|$ be its $j$-th successive minimum. Then
    \[
        \sum_{j=1}^n |m_j| = \covol(L).
    \]
\end{theorem}

\begin{corollary}[Mahler]
\label{cor:mahler}
Let $L$ be a lattice.
Then for the first successive minimum $|m_1|$ we have $|m_1| \leq \frac{\covol(L)}{n}$, i.e., there exists a non-zero lattice vector in $L$ of norm at most $\frac{\covol(L)}{n}$.
\end{corollary}

\section{BIKE from a Lattice Viewpoint}
\label{sect:BIKE_Lattices}

In the BIKE cryptosystem we have a similar problem as in Thue's Lemma \ref{lemma:Thue} from Section \ref{sect:real_lattices}, but in the polynomial setting.
We want to solve
\[
    h_1 - h h_2 \equiv 0 \mod x^r - 1,
\]
for $h_1$ and $h_2$ which are sparse.
This problem has been studied as a problem of rational reconstruction (without the lattice viewpoint) in \cite{BardetDLO16}.
Again we can build a lattice
\[
    L = \{ (h_1, h_2) : h_1 - h h_2 \equiv 0 \mod x^r - 1 \}.
\]
Note that this lattice can be constructed solely from the public key.
We will propose a new algorithm to get weak keys in a similar way as in \cite{BardetDLO16} by using lattice reduction and adding a possibility for using brute force.
In practice our algorithm improves only very slightly upon the algorithm proposed before, but it gives a more unified perspective on it, especially since lattice reduction for polynomial rings over finite fields is computationally easy in general and not just for the $2$-dimensional case as for the integers.

We can even write down an explicit basis for our lattice as we show in the following lemma.
\begin{lemma}
Let $L = \{ (h_1, h_2) : h_1 - h h_2 \equiv 0 \mod x^r - 1 \}$.
Then the vectors
\[
    (x^r - 1, 0), (h, 1)
\]
form a basis of $L$.
\end{lemma}
\begin{proof}
It is obvious that both vectors are in $L$.
To show that they form a basis take $(h_1, h_2) \in L$.
Then also
\[
    (h_1', 0) = (h_1, h_2) - h_2  (h, 1) \in L.
\]
But then
\[
    h_1' \equiv 0 \mod x^r - 1,
\]
and consequently we can express $(h_1, h_2)$ as a linear combination of $(x^r - 1, 0)$ and $(h, 1)$ with coefficients in $\F_q[x]$.
\end{proof}

Note that a similar result also holds for the case of the integers from Section \ref{sect:real_lattices}.
It is now easy to compute the determinant of $L$.
\begin{corollary}
\label{cor:basis_lattice}
    Let $L = \{ (h_1, h_2) : h_1 - h h_2 \equiv 0 \mod x^r - 1 \}$.
    Then the determinant of $L$ is $x^r - 1$.
\end{corollary}
\begin{proof}
    The determinant of 
    \[
        \begin{pmatrix}
        x^r - 1 & 0\\
        h & 1
    \end{pmatrix}
    \]
    is $x^r - 1$.
\end{proof}
From Corollary \ref{cor:basis_lattice} we get the following result for rational reconstruction, similar to Thue's Lemma \ref{lemma:Thue}.

\begin{corollary}
\label{cor:shortest_vector_bike_lattice}
    Let $L = \{ (h_1, h_2) : h_1 - h h_2 \equiv 0 \mod x^r - 1 \}$.
    There exists $a, b \in L$ with $|a| + |b| = r$.
    In particular there exists $a \in L$ with $|a| \leq \frac{r-1}{2}$.
\end{corollary}

\begin{proof}
    We have $\covol(L) = r$ from the previous corollary.
    By applying Theorem \ref{thm:mahler} we get the vectors $a, b$ as desired.
\end{proof}
We could have also proven the statement using Lemma \ref{lemma:index polynomial lattice}.
In fact, for $L$ as above we can compute $a$ and $b$ using Lenstra's algorithm.
We get the following result from Lenstra's Corollary \ref{cor:lenstra complexity}.
\begin{lemma}
\label{lemma:successive_minima_complexity_bike}
    Let $L = \{ (h_1, h_2) : h_1 - h h_2 \equiv 0 \mod x^r - 1 \}$.
    Then one can find a reduced basis and hence also the successive minima of $L$ in time $O(r^2)$.
\end{lemma}

Now, let $h_1$ and $h_2$ be the actual private keys.
In case that $|(h_1, h_2)| < \frac{r}{2}$, using Lemma \ref{lemma:successive_minima_complexity_bike}, we recover essentially the attack from \cite{BardetDLO16} since the attack from \cite{BardetDLO16} usually computes the shortest vector $a$ from Corollary \ref{cor:shortest_vector_bike_lattice}.
In \cite{BardetDLO16} the authors also analyzed the probability that the condition on the degrees is satisfied.

We can now slightly extend this attack using brute force.
It is worth mentioning that in the attack from \cite{BardetDLO16} this is not possible since they only found a shortest vector, but no reduced basis.
As long as the degrees of the private key are not too far from the degrees of the reduced basis we can compute the reduced basis and then try to find the private key as described in Algorithm \ref{alg:weak_keys}.

\begin{algorithm}[H]
\caption{Weak Keys via Lattice Reduction}\label{Weak Key Algorithm}
\label{alg:weak_keys}
\textbf{Input:}
\begin{itemize}
    \item Public Key $h$ for BIKE
    \item A bound $B$ on the degree we use for brute force.
\end{itemize}
\textbf{Step 1:}
Compute a reduced basis $a = (a_{1}, a_{2}), b = (b_{1}, b_{2})$ for the lattice $L = \{ (h_1, h_2) : h_1 - h h_2 \equiv 0 \mod x^r - 1 \}$ using Lenstra's algorithm.\\
\textbf{Step 2:}
\begin{algorithmic}
\For{$\mu_1, \mu_2 \in \F_2[x]_{<B}$}
\State
Compute $(h_1', h_2') = \mu_1 (a_{1}, a_{2}) + \mu_2 (b_{1}, b_{2})$.
\If {$h_1'$ and $h_2'$ are sparse}
\State \Return $(h_1', h_2')$.
\EndIf
\EndFor
\end{algorithmic}
\end{algorithm}

However, there is a slight difference in the weak keys we get.
In \cite{BardetDLO16} they have the condition $\deg(h_1) + \deg(h_2) < r$, whereas we have the condition $|(h_1, h_2)| = \max \{\deg(h_1), \deg(h_2)\} < \frac{r}{2} + B$.
Some of the weak keys from \cite{BardetDLO16} we will not recover with Algorithm \ref{alg:weak_keys}, and on the other hand some of our weak keys will not be recovered by the algorithm in \cite{BardetDLO16}.
Combining both approaches enables us to check for more (however not many more) weak keys than in the original paper.
Using calculations similar to those in \cite{BardetDLO16} we expect to gain at most a few bits upon the original attack for reasonable values on the bound $B$ of the degree when using brute force.
One can also apply to our approach the same improvements as in \cite{BardetDLO16} using group actions which we have not considered here.
In our case the analysis on the success probability follows along similar lines as in \cite{BardetDLO16}.

One problem on which the hardness of BIKE rests is the following.
\begin{problem}
\label{problem:sparse_norm}
    How hard is it to find short vectors in the lattice $L$ where the norm is defined by $|(h_1, h_2)| = \wt(h_1) + \wt(h_2)$.
\end{problem}
However, while being a norm over $\F_q$ endowed with the trivial absolute value, this is not a norm over $\F_q(T)$ endowed with the absolute value arising from the degree and the natural norm on the lattice is quite far from the norm in Problem \ref{problem:sparse_norm}.
Are there better algorithms to tackle this problem?

\section{Conclusion and Open Questions}
\label{sect:conclusion}

In this work we recast the rational reconstruction problem for key recovery of the BIKE cryptosystem as a problem of lattices over polynomial rings.
Two questions that immediately arise are whether taking Mahler's \cite{mahler1941analogue} more general definition of successive minima can give more weak keys and whether they can be computed efficiently.

One of the main open questions is whether there are better key recovery attacks that use a similar approach.
In particular, the lattice constructed in this paper is a rank $2$ lattice and closely related to the Euclidean algorithm and continued fractions.
Is there a way to construct a higher rank lattice that makes it possible to get (many) more keys?

We would like to point out that there have been two problems that started out with rank $2$ lattices (or the extended Euclidean algorithm and continued fractions) and could be improved considerably using higher rank lattices.
One is Wiener's classical attack \cite{wiener1990cryptanalysis} on small private keys using continued fractions and its improvements by Boneh and Durfee \cite{boneh2000cryptanalysis} relying on Coppersmith's attack \cite{coppersmith1997small} which uses the LLL algorithm.
Another one is the decoding of Reed-Solomon codes by Berlekamp and Massey \cite{berlekamp1966nonbinary, massey2003shift} (reinterpreted as continued fractions by Mills \cite{mills1975continued}, Welch and Scholtz \cite{welch1979continued}) and the breakthrough of List Decoding by Guruswami and Sudan \cite{guruswami1998improved} (reinterpreted by Bernstein \cite{Bernstein2008smallheight}, Cohn and Heninger \cite{CohnH15} as an analogue of Coppersmith's theorem).
Can something similar be done for the key recovery of the secret key of BIKE?
At present this question is open and it requires new ideas.

In addition, it might also be worth applying similar methods to HQC \cite{melchor2018hamming}.

\section*{Acknowledgements}

I would like to thank Maxime Bombar for introducing me to the Function Field Approach for quasi-cyclic codes, which underlies cryptosystems like BIKE or HQC.
Furthermore, I would like to thank him for mentioning the work of Bernstein, Cohn and Heninger on the Coppersmith Lattice Approach to List Decoding.
In addition I would like to thank Paolo Santini for bringing \cite{BardetDLO16} to my attention.
I would also like to thank the reviewers for their valuable comments.

This work is supported by the Swiss National Science Foundation under SNSF grant number 212865.

\printbibliography
\end{document}